\documentclass{llncs}

\usepackage{url}
\usepackage{color}
\usepackage{amssymb}
\usepackage{amsmath}
\usepackage{alltt}
\usepackage{graphicx}
\usepackage{color}
\usepackage{stmaryrd}
\usepackage{mathtools}
\usepackage{tikz}
\usetikzlibrary{matrix,arrows}
\usepackage{ascii}
\usepackage{booktabs}
\usepackage{wrapfig}

\usepackage[a4paper, margin=1in]{geometry} 

\newcommand{\efr}[1]{{\fbox{$#1$}}}  
\newcommand{\sph}[1]{\text{\rm\scriptsize [}e\text{\rm\scriptsize ]}}

\newcommand{\entails}{\vdash}
\newcommand{\bentails}{\Vvdash}
\newcommand{\entailsp}[1]{\vdash_{\!\!\!#1}}

\newcommand{\circentails}{\entailsp{{\shortrightarrow}{\shortleftarrow}}}

\newcommand{\CIRC}{{\sf CIRC}}

\newcommand{\simstar}{\stackrel{\scriptstyle{*}}{\sim}}
\newcommand{\simstarplus}{\mathop{{\simstar}{}^+}}

\def\@p#1{\mathrel{\ooalign{\hfil$\mapstochar\mkern
      5mu$\hfil\cr$#1$}}}
\def \pfun	{\@p\rightarrow}

\newcommand{\mycomment}[1]{}

\newenvironment{todo}{\bigskip\hrule\medskip\noindent}{\medskip\hrule\bigskip}
\begin{document}

\title{Circular Induction}

\author{
Dorel Lucanu\inst{1} 
\and
Grigore Ro\c{s}u\inst{2}
\and
Eugen Goriac\inst{1}
\and
Georgiana Caltais\inst{1}
}
\institute{
Faculty of Computer Science\\
Alexandru Ioan Cuza University, Ia\c{s}i, Romania,
\email{[dlucanu,egoriac]@info.uaic.ro}
\and
Department of Computer Science\\
University of Illinois at Urbana-Champaign, USA,
\email{grosu@illinois.edu}
}

\maketitle

\begin{abstract}
The Circularity Principle was successfully applied for developing a coinductive proving technique, known as circular coinduction. In this paper we show that the same principle can be used to develop an inductive proving technique. A main advantage of this uniform approach is that the two proving techniques can be easily combined during the verification process. Circular induction is simple, flexible, generic, and therefore it is a good candidate framework for combining different proving schemes into a competitive tool. We exhibit this potential by presenting how the circular induction is implemented in CIRC, a prover built around the Circularity Principle.
\\[1ex]
\textbf{Disclaimer.} This paper was written in 2010, at the time the CIRC prover was developed, and the main body reflects the state of the work and of the prover as of that date. For this arXiv technical report, only the related-work discussion (Section \ref{sec:related-work}) and the concluding section have been revised: Section \ref{sec:related-work} has been extended to situate circular induction within the cyclic-proof and infinite-descent literature that has appeared or matured since 2010. No other part of the paper—its definitions, results, proofs, examples, or implementation description—has been modified, and the technical content should be read as a 2010 contribution. References to developments after 2010 appear only in the updated related-work section.
\end{abstract}

\pagestyle{plain}

\section{Introduction}
\label{sec:intro}

Circular coinduction~\cite{DBLP:conf/icfem/LucanuR09} is a coinductive proving technique for behavioral properties. By a behavioral property we mean a  property which can be experimentally  evaluated.  The soundness of the circular coinduction can be explained, among other interpretations, by a graph whose equalities from nodes can be regarded as lemmas inferred in order to prove the original task,
and the graph itself as a dependence relation among these lemmas; one can take
all these and produce a parallel proof. We exemplify this technique proving a very simple property over streams (infinite lists). Let {\it zeros} be the streams of 0's (= $0:0:0:\ldots$), {\it ones} the stream of 1's, {\it oz} the stream $0:1:0:1\ldots$, {\it zo} the stream $1:0:1:0\ldots$, and  let $\it zip(S,S')$ be the operation which is zipping two streams. The formal definitions for these streams are given by the following corecursive equations:
\\[0.5ex]
\centerline{$
\begin{aligned}
&\it zeros = 0:zeros\\[-0.5ex]
&\it ones = 1: ones
\end{aligned}
\qquad
\it zip(a:S,S') = a:zip(S',S)
\qquad
\begin{aligned}
&\it zo = 0: oz\\[-0.5ex]
&\it oz = 1 : zo
\end{aligned}
$}\\[0.5ex]
We consider the destructors {\it hd} (head) and {\it tl} (tail) defined by $\it hd(a:S)=a$ and $\it tl(a:S)=S$. The two destructors define two {\it derivatives} $\it hd(*{:}Stream)$ and $\it tl(*{:}Stream)$, which can be seen as equation transformers: an equation $S=S'$ over streams is transformed into $\it hd(S)=hd(S')$ and $\it tl(S)=tl(S')$, respectively. The proof tree of the property $\it zip(zeros, ones)=zo$ by circular coinduction is represented in Fig.~\ref{fig:graph-cc}. The head derivative produces a new property (equation), which can be immediately proved using the definition of the operations. The tail derivative produces a new lemma, which is processed in the same way.
Again, the property produced by the head derivative is immediately proved. The tail derivative produces exactly the initial property, so there is no reason to continue because we ended into a circularity.  For the general case, the circularity is modulo a substitution and/or special contexts.
\begin{figure}
\centering
\vspace*{-6ex}
\hspace*{-2ex}
\begin{tikzpicture}
\matrix (m) [matrix of math nodes, row sep=5em,
column sep=-5em, text height=1.5ex, text depth=0.25ex]
{ 
& \it zip(zeros, ones)=zo & & \\
\begin{array}{c}
\it hd(zip(zeros, ones))=hd(zo)\\
\it 0 = 0~\checkmark
\end{array} & &
\begin{array}{c}
\it tl(zip(zeros, ones))=tl(zo)\\
\it zip(ones, zeros)=oz\\
\end{array} & \\
 & 
\begin{array}{c}
\it  hd(zip(ones, zeros))=hd(oz)\\
\it 1 = 1 \checkmark
\end{array}& &\hspace{-1ex} 
\begin{array}{c}
\\[-1ex]
\it  tl(zip(ones, zeros))=tl(oz)\\
\it  zip(zeros, ones)=zo
\end{array} & \\
};
\path[->]
(m-1-2) edge  node [above] {\it hd}  (m-2-1)
(m-1-2) edge  node [auto] {\it tl}  (m-2-3)
(m-2-3) edge  node [above] {\it hd}  (m-3-2)
(m-2-3) edge  node [auto] {\it tl}  (m-3-4)
(m-3-4) edge [bend left=-40] node[above] {} (m-1-2);
\end{tikzpicture}
\caption{\label{fig:graph-cc}Intuitive proof by circular induction} 
\end{figure}

In this paper we exploit this viewpoint in order to develop a similar inductive proving technique.
We start to explain this technique using a simple example.
Let us consider the following two definitions for even, one iterative and the other one mutual  (using odd):\footnote{This example is from Induction Challenge Problems site: \url{http://www.cs.nott.ac.uk/~lad/research/challenges/}}\\[0.5ex]
\centerline{$
\begin{aligned}
&\it even(0)=true\\[-1ex]
&\it even(s(0))=false\\[-1ex]
&\it even(s(s(N)))=even(N)
\end{aligned}
~~
\begin{aligned}
&\it evenm(0)=true\\
&\it evenm(s(N))=oddm(N)
\end{aligned}
~~
\begin{aligned}
&\it oddm(0)=false\\
&\it oddm(s(N))=evenm(N)
\end{aligned}
$}\\[0.5ex]

\begin{figure}
\centering
\begin{tikzpicture}
\matrix (m) [matrix of math nodes, row sep=5em,
column sep=-3.5em, text height=1.5ex, text depth=0.25ex]
{ 
& \it even(n)=evenm(n) & & \\
\begin{array}{c}
\it even(0)=evenm(0)\\
\it true = true~\checkmark
\end{array} & &
\begin{array}{c}
\it even(s(n_{\rm 1})) = evenm(s(n_{\rm 1}))\\
\it even(s(n_{\rm 1}))=oddm(n_{\rm 1})\\
\end{array} & \\
 & 
\begin{array}{c}
\it even(s(0))=oddm(0)\\
\it false = false \checkmark
\end{array}& & 
\begin{array}{c}
\\[-1ex]
\it even(s(s(n_{\rm 2})\!)\!) = oddm(s(n_{\rm 2})\!)\\
\it even(n_{\rm 2})=evenm(n_{\rm 2})
\end{array} & \\
};
\path[->]
(m-1-2) edge  node [above] {0}  (m-2-1)
(m-1-2) edge  node [auto] {$s(\_)$}  (m-2-3)
(m-2-3) edge  node [above] {0}  (m-3-2)
(m-2-3) edge  node [auto] {$s(\_)$}  (m-3-4)
(m-3-4) edge [bend left=-40] node[above] {$n_{\rm 2}\sim n$} (m-1-2);
\end{tikzpicture}
\caption{\label{fig:graph}Intuitive proof by circular induction} 
\end{figure}

\noindent
The proof of the conjecture $(\forall N)\it even(N)=evenm(N)$ is intuitively represented by the graph in Figure \ref{fig:graph}. The constructors of naturals, $0$ and $s(\_)$, can be seen as properties (here equalities) transformers, by applying an appropriate substitution over an inductive variable. These transformers obtained from constructors are called \emph{derivatives}. The initial goal is transformed by derivatives into two new proof obligations: $\it even(0)=evenm(0)$ and $\it even(s(n_{\rm 1})) = evenm(s(n_{\rm 1}))$.The former is reduced using the axioms from the specification and the latter is simplified to a simpler goal; this is then derived into two new proof obligations: $\it even(s(0))=oddm(0)$ and 
$\it even(s(s(n_{\rm 2})\!)\!) = oddm(s(n_{\rm 2})\!)$. Again the former is reduced using the axioms; the latter is very "similar" to initial goal (the similarity is expressed by $n_2\sim n$ meaning that "$n_2$ is an incarnation of $n$ and therefore satisfies the same properties as $n$), so there is no reason to continue the proving process, and we may conclude that property holds.
The composition of the derivatives produces \emph{experiments}. We may say that the property above was proved using only three such derivations because we ended up into a cycle. It is easy to see now that this proof technique is very similar to the one applied for circular coinduction.

In this paper we formalize the above exemplified proof technique following the line from~\cite{rosu-lucanu-2009-calco,DBLP:conf/icfem/LucanuR09}. We introduce notions like derivative, experiment, freezing operator, special hypothesis  suitable for inductive properties, and we show that the circularity principle holds for the new framework. Then we define the circular induction proof calculus based on this principle. Since both circular coinduction and circular induction are presented at a generic level, parameterized in the basic entailment relation, many proofs from circular coinduction are reused for circular induction by duality. This strengthens our belief that the circularity principle is a foundation for combining different proving paradigms into a powerful prover. We exemplify that by showing how the two mentioned proving techniques  are implemented in {\CIRC}. 
We present in Figure \ref{fig:duality} the correspondence between the similar concepts, which are expressing in fact the duality between the two paradigms. We illustrate the concepts  using streams (for circular coinduction)  and natural numbers (for circular induction), as canonical examples for two approaches.

\begin{figure}
\centering
\vspace*{-3ex}
\begin{tabular}{|l|l|l|}
\hline
 & 
Circular coinduction 
& 
Circular induction \\
\hline
sorts
& 
hidden (coinductive), 
&
inductive (initial),  \\
&
 visible (non-coinductive)
& non-inductive (loose)\\
\hline
derivatives $\Delta$ 
&
destructor (context)
& 
constructor (substitution)\\
 & 
${\it hd}(*{:}{\it Stream})$, ${\it tl}(*{:}{\it Stream})$ 
&
$0{.}{\it Nat}$, $s(N{:}{\it Nat})$\\
\hline
derived equation
&
$\delta[t]=\delta[t']$
&
$\delta_c[t]=\delta_c[t']$\\
&
$\it hd(t)=hd(t')$, $\it tl(t)=tl(t')$
&
$t[0/y{:}{\it Nat}]=t'[0/y{:}{\it Nat}]$, \\
&
&
$t[s(y{:}{\it Nat})/y{:}{\it Nat}]=t'[s(y{:}{\it Nat})/y{:}{\it Nat}]$\\
\hline
$\Delta$-experiment 
& 
visible $\Delta$-context
&
ground $\Delta$-substitution\\
\hline
freezing 
&
at top\qquad ~inhibits
&
at bottom\qquad\qquad inhibits\\
&
$t\mapsto \efr{t}$ \qquad contexts
&
$N{:}{\it Nat}\mapsto N{.}{\it Nat}$\quad ~substitutions\\
\hline
frozen equation
&
$\efr{t}=\efr{t'}$
& 
$t[y{.}s/y{:}s]=t'[y{.}s/y{:}s]$\\
\hline
\end{tabular}
\vspace*{-2ex}
\caption{\label{fig:duality}Duality between circular coinduction and circular induction}
\vspace*{-3ex}
\end{figure}

We regard circular induction as a rather complementary inductive proving technique which can contribute to the development of more powerful automated provers. The main advantages of the circular induction are its simplicity - only two inference rules, no explicit well-ordering, no algorithm for computing cover/test sets, flexibility - it can be combined with other proving techniques, e.g., circular coinduction, or disproving techniques using simplification rules~\cite{synasc09} or decision procedures, generality - it is parametric in the basic entailment relation, which allows using circular induction as framework suitable to combine more proving techniques. 

The structure of the paper is as follows. In Section~\ref{sec:prelim}
we present some preliminary notions used throughout the paper, including a brief recall of the circular coinduction in order to facilitate the comparison of the two proof techniques.
Section~\ref{sec:ind-conseq} shows how inductive consequences
can be regarded as behavioral properties.
Section~\ref{sec:ci} is dedicated to presenting the circular induction
principle and the circular induction proof calculus. In Section~\ref{sec:impl} we show how circular induction is implemented in {\CIRC} and how the two proof techniques, circular coinduction and circular induction may work together. 
Section \ref{sec:related-work} discusses the related work and Section~\ref{sec:concl} concludes the paper and outlines future work.

\section{Preliminaries}
\label{sec:prelim}

We assume the reader familiar with the basics of many sorted algebraic
specifications and only briefly recall our notation.  An algebraic
specification, or simply a {\em specification}, is a triple
$(S,\Sigma,E)$, where $S$ is a set of {\em sorts},  $\Sigma$ is a
{\em $(S^*\times S)$-signature} and $E$ is a set of
{\em $\Sigma$-equations} of the form
$(\forall X)\,t=t'{\tt ~if~}\land_{i\in I} u_i=v_i$ with
$t$, $t'$, $u_i$, and $v_i$ $\Sigma$-terms with variables in $X$,
$i=0,\ldots,n$; the pair $t$, $t'$ and each pair $u_i$, $v_i$ for
each $i\in I$, have, respectively, the same sort.  If the sort of $t$
and $t'$ is $s$ we may say that the sort of the equation is also $s$.
When $i=0$ we call the equation unconditional and omit the condition
(i.e., write it $(\forall X)\,t=t'$).  When $X=\emptyset$ we drop the
quantifier and call the equation {\em ground}.
An \textbf{equational transformer} is a mapping $\delta$ which maps a $\Sigma$-equation $e$ into a set of $\Sigma$-equations $\delta[e]$. If $\Delta$ is a set of equational transformers and $E$ a set of equations, then
$\Delta[e]=\bigcup_{\delta\in\Delta}\delta[e]$ and
$\Delta[E]=\bigcup_{e\in E}\Delta[e]$. The following introduces the notion of basic entailment relation used by both circular proof techniques.

\begin{definition}
\label{dfn:ent}
 Let $(S,(\Sigma,\Delta))$ be a signature with equational transformers. 
A \textbf{$\Delta$-contextual entailment system} is an (infix)
relation $\entails$ between sets of equations and equations,
with:
\begin{enumerate}
\item 
\textbf{(reflexivity)} $\{e\}\entails e$;
\item 
\textbf{(monotonicity)} If $E_1\supseteq E_2$ and $E_2\entails e$ then
$E_1\entails e$;
\item 
\textbf{(transitivity)} If $E_1 \entails E_2$ and
$E_2\entails e$ then $E_1 \entails e$; 
\item 
and
\textbf{($\Delta$-congruence)} If $E\entails e$ then 
$E\entails \Delta[e]$.
\end{enumerate} 

Above, $E$, $E_1$, $E_2$ range over sets of equations and $e$
over equations; we tacitly extended $\entails$ to 
sets of equations as expected:  $E_1\entails E_2$ iff
$E_1\entails e$ for any $e\in E_2$.
\end{definition}

 If ${\cal B}=(S,(\Sigma,\Delta),E)$ is a specification with equational transformers, we often write ${\cal B}\entails e$ for $E\entails e$. 

\subsection{Behavioral Coinductive Specifications and Circular Coinduction}

By behavioral properties we mean those properties which can be proved using experiments: a behavioral property $e$ holds if and only if it cannot be disproved by experiments, i.e., for each experiment $C$, the instance of $e$ corresponding to $C$, usually written as $C[e]$, holds. In this section we briefly recall the notion of behavioral specification and circular coinduction proof system; this will help us to compare and emphasize the duality of the two proof systems. For more details, the reader is invited to read \cite{rosu-lucanu-2009-calco,DBLP:conf/icfem/LucanuR09}. 

A \textbf{$\Sigma$-context} for the sort $s$ is a term $\delta[*{:}s]$ having a single occurrence of the distinguished variable $*{:}s$. If $\delta[*{:}s]$ is a context and $t$ a $\Sigma$-term of sort $s$, then $\delta[t]$ denotes the term $t$ where the variable $*{:}s$ is substituted with $t$. A $\Sigma$-context can be seen as an equational transformer, where $\delta[e]$ is the equation $(\forall X)\,\delta[t]=\delta[t']{\tt ~if~}\land_{i\in I} u_i=v_i$.
A behavioral coinductive signature $(S,(\Sigma,\Delta))$ is a signature $(S,\Sigma)$ together with a set $\Delta$ of $\Sigma$-contexts viewed as equational transformers, called \textbf{derivatives}. The derivatives $\delta[*{:}s]\in\Delta$ are also known as \textit{destructors}. A sort $s$ occurring into a derivative $\delta[*{:}s]$ is called \textbf{hidden}; the other sorts are called \textbf{visible}. Therefore, the sorts $S$ are partitioned into two subsets: the subset of hidden sorts $H$ and the subset of visible sorts $V$. The set of $\Delta$-\textbf{experiments} is inductively defined as follows: 1) any visible derivative $\delta[*{:}h]$ is an experiment for $h\in H$, and 2) if $C[*{:}h]$ is an experiment for $h$ and $\delta[*{:}h']$ is a derivative of sort $h$ (and for $h'$), then $C[\delta[*{:}h']]$ is an experiment for $h'$.
A \textbf{behavioral coinductive specification} $\cal B$ is a triple $(S,(\Sigma,\Delta),E)$, where $(S,(\Sigma,\Delta))$ is a behavioral coinductive signature and $E$ is a set of $\Sigma$-equations. In the rest of this section, $\cal B$ denotes a behavioral coinductive specification as above and $\entails$ a $\Delta$-contextual entailment relation.

\begin{definition}
\label{dfn:beh-sat}
$\mathcal B$ \textbf{behaviorally satisfies}
equation $e$, written $\mathcal B\bentails e$, iff:
$\mathcal B\entails e$ if $e$ is visible, and
$\mathcal B\entails C[e]$ for each appropriate experiment $C$ if $e$
is hidden.
\end{definition}

\begin{example}
\label{ex:streams}
The behavioral coinductive  specification \textit{STREAM} of streams
discussed in Section \ref{sec:intro}, consists of a hidden sort {\it Stream} for streams, a visible sort {\it Bit}  for data elements, a set of operations among which we point out the head $\it hd:Stream\to Bit$ and $\it tl:Stream\to Stream$, and a set of equations defining the other operations ({\it zeros}, {\it ones}, {\it zip} \ldots) in terms of head and tail. The set $\Delta$ of derivatives is $\it\{hd(*{:}Stream),tl(*{:}Stream)\}$ and
the experiments one can perform on streams are therefore contexts of
the form $\it hd(tl^{\it i}(*{:}Stream))$, where $ i\ge 0$. The  basic entailment relation is given by the conventional equational reasoning; in other words,
$\it STREAM\entails e$ iff $\it e$ is derivable using
equational reasoning from the equations in $\it STREAM$.
\end{example}

A key notion in the definition of circular
coinduction proof system is that of a ``frozen'' equation.  The motivation
underlying frozen equations is that they structurally inhibit their
use underneath proper contexts; because of that, they will allow us to
capture the informal notion of ``circular behavior''. 
Formally, the signature $\Sigma$ is extended with a new sort
{\it Frozen} and a new operation $\efr{\text{-}}: s\to {\it Frozen}$
for each sort $s$.  If $t$ is a term, then we call $\efr{t}$ the
{\em frozen (form of) $t$}. 
If $e$ is an equation $(\forall X)\,t=t'{\tt ~if~}c$, then we let
$\efr{e}$ be the {\em frozen equation}
$(\forall X)\,\efr{t}=\efr{t'}{\tt ~if~}c$; note that the condition
$c$ stays unfrozen, but recall that we only assume visible conditions.
The equations over the original signature
$\Sigma$ are called {\em unfrozen equations}.

\begin{definition}\label{cc:entail-with-freezing}
A \textbf{$\Delta$-contextual entailment system with freezing}
is a $\Delta$-contextual entailment system such
that:
\begin{quote}
\begin{enumerate}
\item[(A1)]\label{assmptn1} $E\cup{\cal F} \entails \efr{e}$ \ \ \ iff\ \ \ 
  $E\entails e$;
\item[(A2)]\label{assmptn3} $E\cup{\cal F}\entails{\cal G}$
  \ \ \ implies\ \ \ $E\cup \delta[{\cal F}]\entails\delta[{\cal G}]$
for each $\delta\in\Delta$, 
\end{enumerate}
where $E$ ranges over unfrozen equations, $e$ over visible
unfrozen equations, and $\cal F$ and $\cal G$ over frozen hidden
equations.
\end{quote}
\end{definition}

However, the freezing operator is too restrictive. There are contexts under which is safe to use the frozen hypothesis; these are called \emph{special contexts}. The formal definition is technical and can be found in \cite{DBLP:conf/icfem/LucanuR09}. Obviously, a special context does not include derivatives. For the example of streams, the contexts $\it zip(*{:}Stream,S)$ and $\it zip(S,*{:}Stream)$ are special, whereas $\it odd(*{:}Stream)$ is not special, where $\it odd:Stream\to Stream$ is defined by $\it hd(odd(S))=hd(S)$ and $tl(odd(S))=odd(tl(tl(S)))$ ( a counterexample can be found in \cite{goguen-lin-rosu-2003-wadt,DBLP:conf/icfem/LucanuR09}). Each special context $\gamma[*{:}s]$ defines an equation transformer similar to derivatives.

\begin{definition}\label{dfn:ccind} 
The \textbf{circular coinduction proof system} consists of the following inference rules, represented as conditional reduction rules:
\begin{enumerate}
\item[]
\textsf{[Reduce]}:
$({\mathcal B}, {\mathcal F}, {\mathcal G} \cup \{\efr{e}\})
\Rightarrow
({\mathcal B}, {\mathcal F}, {\mathcal G})
\texttt{\;if\;} {\mathcal B} \cup {\mathcal F} \ {\entails}\ e
$
\item[]
\textsf{[Derive]}:
$({\mathcal B}, {\mathcal F}, {\mathcal G} \cup \{\efr{e}\})
\Rightarrow
({\mathcal B}, {\mathcal F} \cup \{\efr{e}\}\cup \Gamma[e], {\mathcal G} \cup \{\efr{\Delta[e]}\})$\\
\verb##\hspace{39pt}
$\texttt{\;if\;} {\mathcal B} \cup {\mathcal F}\ {\not \entails}\ e \,\land\, e \textrm{ is
hidden}$
\end{enumerate}
where $\cal F$ a set of frozen hypotheses, $\cal G$ is a set of goals represented also as frozen equations, and $\Gamma$ is a set of special contexts. 
\end{definition}

The rule \textsf{[Reduce]} says that if the current goal is a $\entails$--consequence
of ${\mathcal B} \cup {\mathcal F}$ then $\efr{e}$ is dropped. 
When the current goal $e$ is hidden and it is not a
$\entails$--consequence, it and the equations defined by the special contexts over it are added 
to the specification and its derivatives to the set of goals using the rule \textsf{[Derive]}.
The proof of $\it zip(zeros, ones)=zo$ given in Section~\ref{sec:intro} is obtained with the above system, where all the equations should be considered frozen. Without freezing the added hypotheses, we can wrongly deduce, e.g., that $\it zeros = oz$.

\subsection{Inductive Specifications}

A \textit{signature with constructors} is a pair $(S,(\Sigma,\Sigma^{\it ctor}))$, where $(S,\Sigma)$ is a many sorted signature and $\Sigma^{\it ctor}\subseteq \Sigma$ is the subsignature of \textit{constructors}.
A sort $s$ is called \textit{inductive sort} if there is a constructor of sort $s$. 
An  \textit{inductive variable} is a variable of inductive sort. 
A \textit{goal (conjecture)} is a conditional equation written as $(\forall Y)(\forall Z)t=t'$~{\tt if}~{\it cond}, where $Y$ is a set of inductive variables.
If $\theta:Y\to T_{\Sigma^{\it ctor}}$ is a constructor ground substitution, then $\theta[e]$ denotes the equation  $(\forall\emptyset)(\forall Z)t=t'$~{\tt if}~{\it cond}.
A goal $e$ is an (equational) \textbf{inductive consequence} of the specification $(S,(\Sigma,\Sigma^{\it ctor}),E)$, if $(S,\Sigma,E)$ satisfies in the equational deduction system all constructor ground instances $\theta[e]$. More details about inductive consequences can be found, e.g., in \cite{ci:bundy}.

In order to exhibit the concepts and results presented in this paper, we consider the following examples. 

\begin{example}\label{ex:nat-sum}
 \textit{(Sum of natural numbers)}\\
The specification {\it NATSUM}, defining natural numbers, consists of a sort {\it Nat}, two constructors $0:{}\to{\it Nat}$ and $s:{\it Nat}\to{\it Nat}$, an operator $\it sum :  Nat~Nat\to Nat$ defined by the equations $\it sum(M, 0)=0$, $\it sum(M, s(N))=s(sum(M,N))$. We will exhibit how the commutativity and associativity of {\it sum} can be automatically proved with the basic circular induction proof system; two new lemmas are discovered during the proof process: $\it sum(0,M)=M$ and $sum(s(M),N)=s(sum(M,N))$. The goal expressing the associativity is 
given by the unconditional equation $\it (\forall P)(\forall M,N)sum(M, sum(N, P)) = sum(sum(M, N), P)$, where only $P$ is viewed as an inductive variable for this goal.
In order to exemplify how non-linear goals are handled, we use a {\it max} operator defined by $\it max(M,0)=M$, $\it max(0,s(N))=s(N)$, $\it max(s(M),s(N))=s(max(M,N))$, and we show that the non-linear goal  $\it max(N,N)=N$ is automatically proved in exactly one derivation step.
\end{example}

\begin{example}\label{ex:tree}
 \textit{(Generic trees)}\\
This example exhibits inductive properties over generic data types defined using the mutual recursion and the case when a constructor could have more than one inductive parameter. The specification {\it TREE} includes the inductive definition of trees built over a parameter data type {\it Elt} and where a node has a list of children. The lists are defined using a sort {\it TList} and the constructors\\[0.5ex]
\centerline{
$\it nil:{}\to TList$\qquad $[\_]:\it Tree \to TList$\qquad $\_{;}\_{}:\it TList~TList\to TList$
}\\[0.5ex] 
Note that the concatenation is defined as a constructor; it replaces the usual constructor $\it cons: Tree~TList\to TList$. We made this choice because we want to exhibit how the nonlinear conjectures are handled in more complex examples. The trees are defined using the sort {\it Tree} and the constructors\\[0.5ex]
\centerline{
$\it [\_]:Elt\to Tree$\qquad\quad $\it tr:Elt~TList\to Tree$
}\\[0.5ex]
We further consider the mirroring operation defined by\\[0.5ex]
\centerline{
$\begin{aligned}
&\it mirror(nil)=nil\\[-1ex]
&\it mirror([T])=[mirror(T)]\\[-1ex]
&\it mirror(L_1;L_2)=mirror(L_2);mirror(L_1)
\end{aligned}$
$\begin{aligned}
&\it  mirror([E]) = [E] \\[-0.5ex]
&\it mirror(tr(E, L)) = tr(E, mirror(L)) 
\end{aligned}$
}\\[0.5ex]
We show how the goal $\it mirror(mirror(L{:}TList))=L{:}TList$ is automatically proved by circular induction.  Note that the property can only be proved together with $\it mirror(mirror(T{:}Tree))=T{:}Tree$, which is automatically discovered using the help of the first goal. 
\end{example}

\begin{example}\label{ex:fib}
 \textit{(Fibonacci sequence)}\\
The specification {\it FIB} includes the definition of the Fibonacci sequence. The definition of the
natural numbers consists of a sort {\it Nat} and two constructors $0:{}\to{\it Nat}$ and $s:{\it Nat}\to{\it Nat}$. We also consider an associative and commutative operator $\it \_+\_ :  Nat~Nat\to Nat$ defined by the equations $M+ 0=0$, $\it M+ s(N))=s(M+N)$. The definition of Fibonacci sequence is given by\\[0.5ex]
\centerline{
${\it fib}(0)=0$\quad ${\it fib}(s(0)) = s(0)$\quad $\it fib(s(s(N)))=fib(s(N))+fib(N)$
}\\[0.5ex]
 and the definition of the number of calls of $+$ is given by\\[0.5ex]
\centerline{
${\it np}(0)={\it np}(s(0))=0$\qquad\quad ${\it np(s(s(N)))=np(s(N)) + np(N)}+s(0)$
}\\[0.5ex]
We exhibit how the goal $\sum_{k=0}^N{\it fib}(k)=np(s(N))$ can be proved with the enhanced circular induction proof system implemented in {\CIRC}.
\end{example}

\section{Inductive Consequences are Behavioral Properties}
\label{sec:ind-conseq}

 In this section we show that the inductive consequences can be viewed as behavioral properties, i.e., we can define notions like experiments and derivatives, which satisfy the same properties as for the coinductive case, and define the inductive satisfaction similarly to behavioral coinductive satisfaction.

\begin{definition}
Let $(S,(\Sigma,\Sigma^{\it ctor}))$ be a signature with constructors. Each constructor $c\in\Sigma^{\it ctor}$ defines a \textbf{derivative} (equational transformer) $\delta_c=\{\delta_{c,y}\}$ as follows: if $e$ is $(\forall Y)(\forall Z)t=t'$~if~{\it cond} then
\begin{itemize}
\item $\delta_{c,y}[e]$ is $(\forall Y\setminus\{y\})(\forall Z)t[c/y]=t'[c/y]$~if~{\it cond}$[c/y]$ for each $y\in Y$ and constructor constant $c$ having the same sort as that of $y$;
\item $\delta_{c,y}[e]$ is $(\forall Y\setminus\{y\}\cup\{y_1,\ldots,y_n\})(\forall Z)t[u/y]=t'[u/y]$~if~{\it cond}$[u/y]$ for each $y\in Y$ of sort $s$ and $c:s_1\ldots s_n\to s$, where $u=c(y_1,\ldots,y_n)$ and $y_i$ is a fresh variable of sort $s_i$ for $i=1,\ldots,n$; in this case we write $y_i\sim y$ for each $i$ such that $s_i=s$ ($y_i$ can be seen as a new \emph{incarnation} of $y$);
\item $\delta_c[e]=\{\delta_{c,y}[e]\mid y\in Y\land {\it sort}(c)={\it sort}(y)\}$ for each $c\in\Sigma^{\it ctor}$.
\item $\Delta=\cup_{c\in \Sigma^{\it ctor}}\delta_c$.
\end{itemize}
\end{definition}

\begin{example}
Let us consider the specification {\it TREE} defined in Example \ref{ex:tree}.
If $e$ is $\it mirror(mirror(T:Tree))=T:Tree$ and $c$ is $\it tr(E{:}Elt, L{:}TList)$, then $\delta_c[\efr{e}]$ consists of $\it mirror(mirror(tr(E,L)))=tr(E,L)$.
If $\efr{e}$ is $\it mirror(mirror(L:TList))=L$ and $c$ is $\it  L_{\rm 1}{:}TList;L_{\rm 2}{:}TList$, then $\delta_c[\efr{e}]$ consists of ${\it mirror(mirror(L_{\rm 1};L_{\rm 2}))}=L_1;L_2$ and $L\sim L_1$ and $L\sim L_2$.
\end{example}

The relation $\sim$ between frozen variables plays also a key role in the definition of the circular induction proof system.

\begin{proposition}\label{prop:exprm}
Let $e$ be an equation $(\forall Y)(\forall Z)t=t'$~if~{\it cond} and let $\theta : Y\to {\cal T}_{\Sigma^{\it ctor}}$ be a ground $\Sigma^{\it ctor}$-substitution. Then there exist certain constructors (not necessarily distinct) $c_1,\ldots, c_n$ and inductive variables $y_1,\ldots,y_n$ such that 
$Y\subseteq\{y_1,\ldots,y_n\}$ and $\theta[e]=\delta_{c_1,y_1}[\ldots\delta_{c_n,y_n}[e]\ldots]$. 
\end{proposition}
\begin{proof}
Let ${\it der}(\theta)[e]$ denote the equation inductively defined as follows:
\begin{itemize}
\item if $Y=\{y\}$ and $\theta(y)$ is a constant $c$, then ${\it der}(\theta)[e]$ is $\delta_{c,y}[e]$;
\item if $Y=\{y\}$ and $\theta(y)=c(u_1,\ldots,u_n)$, then ${\it der}(\theta)[e]$ is ${\it der}(\theta')[\delta_{c,y}[e]]$, where $\theta':\{y_1,\ldots,y_n\}\to{\cal T}_{\Sigma^{\it ctor}}$ and  $\theta'(y_i)=u_i$ for $i=1,\ldots,n$;
\item  if $Y=Y'\cup \{y\}$, then ${\it der}(\theta)[e]$ is ${\it der}(\theta|_{Y})[{\it der}(\theta|_{\{y\}})[e]]$.
\end{itemize}
It is easy to check that $\theta[e]$ and ${\it der}(\theta)[e]$ are the same and ${\it der}(\theta)[e]$ has the shape asked by the conclusion.
\qed
\end{proof}

For instance, if $e$ is $(\forall N){\it max}(N,N)=N$ and $\theta(N)=s(0)$, then $\theta[e]$ is the same with $\delta_0[\delta_s[e]]$.

Proposition~\ref{prop:exprm} says that the ground $\Sigma^{\it ctor}$-substitutions are in fact {experiments} and can be inductively defined in a similar way to the coinductive case.
\begin{definition}
An \textbf{experiment} for the inductive variable $y$ is defined as follows: 
\begin{itemize}
\item if $c\in \Sigma^{\it ctor}$ of sort $s$ and $y$ is an inductive variable of sort $s$, then
$\delta_{c,y}[\bullet]$ is an experiment for $y$ (it designates the substitution $y\mapsto c$);
\item if $(c:s_1\ldots s_n\to s)\in \Sigma^{\it ctor}$,  $y$ is an inductive variable of sort $s$, and $\theta_i$ is an experiment for $y_i$ of sort $s_i$ (designating the substitution $y_i\mapsto t_i$), then $[\theta_1,\ldots,\theta_n][\delta_{c,y}[\bullet]]$ is an experiment for $y$ (designating the substitution $y\mapsto c(t_1,\ldots,t_n)$;
\item an experiment for a set of inductive variable $Y$ consists of a experiment $\theta$ for each $y\in Y$.
\end{itemize}
The notation $[\theta_1,\ldots,\theta_n][\bullet]$ is a shorthand for $\theta_1[\ldots[\theta_n[\bullet]]\ldots]$ (or an equivalent one because the order we compose $\theta_i$ does not matter).
\end{definition}
It is worth noting to emphasize the duality between the two notions of experiment: for the coinductive case it is a context for a hidden sort, for the inductive case it is a substitution for a set of inductive sorts.
The inductive definition of the experiments and their view as equation transformers are fully exploited by circular induction proof system we define later.
From now on we often write $(S,(\Sigma,\Delta))$ for $(S,(\Sigma,\Sigma^{\it ctor}))$ in order to have a uniform notation and to stress that a signature with constructors defines a set of equational transformers. This double notation is not confusing because if we know $\Sigma^{\it ctor}$ then we can compute $\Delta$ and, reciprocally, if we know $\Delta$ then we can compute $\Sigma^{\it ctor}$ and $[\bullet]$ is a placeholder for equations.
Having defined $\Delta$, it does make sense to consider $\Delta$-contextual entailment systems $\entails$ for the signature with constructors, and define inductive consequences parameterized over these systems, similar it is done for coinductive behavioral consequences.

\begin{definition}
Let ${\cal B}=(S,(\Sigma,\Delta),E)$ be a specification with constructors and let $e$ be an equation $(\forall Y)(\forall Z)t=t'$~if~{\it cond}.
We say that $\cal B$  \textbf{inductively satisfies} $e$, written ${\cal B}\bentails e$, if ${\cal B}\entails \theta[e]$ for each experiment $\theta$ for $Y$.
\end{definition}

\section{Circular Induction}
\label{sec:ci}

Similar to circular coinduction, a key notion for circular induction is that of frozen equation.  If for the case of circular coinduction, the frozen equations inhibit the use of the frozen equations "underneath" proper contexts, for circular induction the frozen equation inhibit their  use "over" proper substitution. The frozen equations help us to capture the informal notion of "circular inductive reasoning" elegantly, rigorously, and generally (modulo a restricted form of equational reasoning).

\begin{definition}\label{def:freezing}
1. A \textbf{frozen form of a variable} $y{:}s$ is a distinguished constant $y{.}s$ of sort $s$.\\
2. A \textbf{frozen form of a goal} $e\equiv(\forall Y)(\forall Z)t=t'$~if~{\it cond} is the equation $\efr{e}\equiv (\forall \emptyset)(\forall Z)t=t'$~if~{\it cond}, where the inductive variables $Y$ are replaced with their frozen forms.\\
3. The set of frozen variables of a goal $e$ is ${\it FVar}(\efr{e})=\{y{.}s\mid y{:}s\in Y\}$.\\
4. If $F$ is a set of goals, then $\efr{F}=\{\efr{e}\mid e\in F\}$ and ${\it FVar}(\efr{F})=\cup_{ e\in F}{\it FVar}(\efr{e})$.
\end{definition}

Let $\delta_c$ be a derivative, $e$ a goal, and $F$ a set of goals . Then we write $\delta_{c,y}[\efr{e}]$ for $\efr{\delta_{c,y}[e]}$, $\delta_c[\efr{e}]$ for $\efr{\delta_c[e]}$, $\Delta[\efr{e}]$ for $\efr{\Delta[e]}$, and $\Delta[\efr{F}]$ for $\efr{\Delta[F]}$.
If $\theta$ is an experiment and $\Theta$ is a set of experiments, then  $\theta[\efr{e}]=\efr{\theta[e]}$, $\theta[\efr{F}]=\{\theta[\efr{e}]\mid e\in F\}$, and $\Theta[\efr{F}]=\cup_{\theta\in\Theta}\theta[\efr{F}]$.

Let $(S,\Sigma,E)\cup\efr{F}\entails\efr{e}$ denote the entailment relation  $(S,\Sigma(Y_F\cup Y_e),E\cup\efr{F})\entails\efr{e}$, where $Y_F$ is the set of the all designated inductive variables of the equations $F$, $Y_e$ is the set of the all designated inductive variables of $e$, and
$\Sigma(Y)$ is the signature $\Sigma$ enriched with the frozen forms of the variables $Y$ (recall that these are constants).
The following is an adaptation of Definition~\ref{cc:entail-with-freezing} for the case of derivatives defined by constructors.
\begin{definition}\label{def:ent-with-freezing}
Let $(S,(\Sigma,\Delta))$ be a signature together with the associated derivatives. 
A \textbf{$\Delta$-contextual entailment system with freezing} is a $\Delta$-contextual entailment $\entails$ satisfying the following two assumptions:\\
(A1) $(S,\Sigma,E)\cup\efr{F}\entails \efr{e}$ iff $(S,\Sigma,E)\entails e$ if $\efr{e}$ is an equation without frozen variables.\\
(A2)  if $(S,\Sigma,E)\cup\efr{F}\entails \efr{G}$ then $(S,\Sigma,E)\cup\delta_c[\efr{F}]\entails\delta_c[\efr{G}]$.
\end{definition}

If $\entails$ is the equational deduction, then the axioms (A1) and (A2) are satisfied. If $e$ has no inductive variables then obviously no axiom from $\efr{F}$ is used in the inference of $\efr{e}$. The proof of a goal $\efr{g}[\delta_c/y.s]\in\delta_c[\efr{G}]$ from $(S,\Sigma,E)\cup\delta_c[\efr{F}]$ is obtained from a proof of $g$ from $(S,\Sigma,E)\cup\efr{F}$ by replacing each equation $e$ from this proof with $e[\delta_c/y.s]$.


\begin{theorem}\textbf{(circular induction principle)}\\
\label{thm:icp}
Let ${\cal B}=(S,(\Sigma,\Delta),E)$ be a specification and $F$ be a set of goals. If ${\cal B}\cup\efr{F}\entails \Delta[\efr{F}]$ then ${\cal B}\bentails F$.
\end{theorem}

\begin{proof}
Since both the circular coinduction and circular induction are defined over entailment relations and dual freezing operators satisfying the same properties, we may reuse the proof of Theorem 2 given in \cite{rosu-lucanu-2009-calco}, by defining a dual definition for contexts and a well-founded order over them. For circular coinduction, contexts are defined in the same way as experiments, but they could have arbitrary sorts, including hidden. Dually, a "context" for circular induction is a $\Delta$-substitution $\theta:Y\to {\cal T}_{\Sigma^{\it ctor}}(Y')$ defined in the same way as the experiments but removing the constraint of being ground. If $\theta:Y\to {\cal T}^{\Sigma^{\it ctor}}(Y')$ and $\theta':Y'\to {\cal T}^{\Sigma^{\it ctor}}(Y'')$ are two contexts for $Y$ and respectively $Y'$, then $\theta'[\theta]:Y\to {\cal T}^{\Sigma^{\it ctor}}(Y'')$ is the context defined by $\theta'[\theta](y)=\theta'(\theta(y))$ for each $y\in Y$, where we implicitly assumed $\theta$ extended over terms. We consider the well-founded order over contexts given by $\theta\succ\theta'$ iff there is $\theta''$ such that $\theta=\theta'[\theta'']$ or $\theta=\theta''[\theta']$. Now, the conclusion of Theorem~\ref{thm:icp} follows by performing the same steps as in the proof of Theorem 2 in \cite{rosu-lucanu-2009-calco}.
\end{proof}

The above proof shows  that the circularity mechanism is abstracted from both implementations, of the entailment relation and of the freezing operator; it depends only on some assumptions these two concepts have to satisfy. However, the duality relationship between the corresponding concepts is not obvious; a direct proof of the above theorem for the case when the basic entailment is equational deduction is tedious and very long.

Moreover, the result above can be extended using the mechanism of the special hypotheses~\cite{DBLP:conf/icfem/LucanuR09}. We introduce first some notations. If $\theta$ is an experiment for $Y$, then $|\theta|=\max\{|\theta(y)|\mid y\in Y\}$, where $|t|$ denotes the depth of the term $t$. Let $\theta^\leq$ denote the set of experiments (for $Y$) $\theta'$ with $|\theta'|\leq|\theta|$. 

\begin{definition}
\label{def:spechyp}
\textbf{(special hypotheses)}
Let ${\cal B}=(S,(\Sigma,\Delta),E)$ be a specification and $F$ a
set of goals.  The \textbf{special-hypothesis closure} of $F$ is the set $F^\leq$ satisfying: 1) ${\cal B}\bentails \theta^\leq[\efr{F}]$ implies ${\cal B}\entails \theta[\efr{F^\leq}]$, 2) $F_1\subseteq F_2$ implies $F^\leq_1\subseteq F^\leq_2$, and  3) $(F^\leq)^\leq=F^\leq$.\\
An equation $e$ is \textbf{special} for $F$ if $e\in F^\leq$.
\end{definition}

The first condition is general enough to include many induction schemes, e.g., strong induction (see Example~\ref{ex:succ}), and the third condition says that $F^\leq$ is maximal.
The following result is similar to Theorem~3 in~\cite{DBLP:conf/icfem/LucanuR09} and shows that the special hypotheses can be used as lemmas for free; its proof is similar to Theorem~\ref{thm:icp}. 

\begin{theorem}\label{thm:icp-ext}
\textbf{(extended inductive circularity principle)}

Let ${\cal B}=(S,(\Sigma,\Delta),E)$ be a specification.
If $F$ is a set of goals such that
$\mathcal B\cup \efr{F^\leq}\entails \Delta[\efr{F}]$,
then ${\cal B}\bentails F^\leq$.
\end{theorem}


\begin{definition}
\label{def:ci}
The \textbf{circular induction proof system} consists of the following inference rules, represented as conditional reduction rules:
\begin{align*}
& {\sf [Reduce]}\\
&\quad ({\cal B},{\cal F},{\cal G}\cup\{\efr{e}\})\Rightarrow ({\cal B},{\cal F},{\cal G})
\quad{\rm~if~}{\cal B}\cup{\cal F}\entails \efr{e}\\
& {\sf [Derive]}\\
&\quad ({\cal B},{\cal F},{\cal G}\cup\{\efr{e}\})\Rightarrow ({\cal B},{\cal F}\cup{\cal H}(\efr{e},{\it FVar}(\Delta[\efr{e}])),{\cal G}\cup\Delta[\efr{e}])\quad {\rm~if~}{\cal B}\cup{\cal F}\not\entails \efr{e}
\end{align*}
where ${\cal H}(\efr{e}, FV)=\{\efr{e}[y'/y]\mid y'\sim y, y'\in FV, y\in{\it FVar}(\efr{e})\}$
\end{definition}

The rule \textsf{[Reduce]} eliminates a goal whenever it can be deduced using the the basic entailment system. The key role is played by \textsf{[Derive]} rule: if a goal cannot be deduced with the initial entailment relation, it and the equations "similar" to it are added as frozen hypotheses and its derived goals are added as new proof obligations.

\begin{example}(\textit{successful examples})\\
Assume  that $\cal B$ is the specification {\it NATSUM} defined in Example \ref{ex:nat-sum} and that $\cal G$ consists of the goal $\it sum(m,n)=sum(n,m)$, where $m$ is the frozen form of $\it M{:}Nat$ and $n$ is the frozen form of $\it N{:}Nat$. We also assume that the basic entailment relation $\entails$ is the equational deduction. The initial $\cal F$ is empty. The derivation tree  of the for proving the above goal is described by the following table:
\begin{center}
\begin{align*}
& \cal F &&\cal G &&\rm Rule\\
\hline
1.
&
\emptyset 
&&
\it sum(m, n)=sum(n, m)
&&
\textsf{[Derive]}\\
\hline
2.
&
\begin{aligned}
  &\it sum(m, n)=sum(n, m)\\[-1ex]
  &\it sum(m, n_{\rm 1})=sum(n_{\rm 1}, m)
\end{aligned}
&&
\begin{aligned}
  &\it sum(m, 0) = sum(0, m)\\[-1ex]
  &\it sum(m, s(n_{\rm 1})) = sum(s(n_{\rm 1}), m)
\end{aligned}
&&
\textsf{[Derive]}\displaybreak[0]\\
\hline
3.
&
\begin{aligned}
  &\it sum(m, n)=sum(n, m)\\[-1ex]
  &\it sum(m, n_{\rm 1})=sum(n_{\rm 1}, m)\\[-1ex]
  &\it m_{\rm 1} = sum(0, m_{\rm 1})
\end{aligned}
&&
\begin{aligned}
  &\it 0 = sum(0, 0) \\[-1ex]
  &\it s(m_{\rm 1}) = sum(0, s(m_{\rm 1})) \\[-1ex]
  &\it sum(m, s(n_{\rm 1})) = sum(s(n_{\rm 1}), m)
\end{aligned}
&&
\textsf{[Reduce]}\displaybreak[0]\\
\hline
4.
&
\begin{aligned}
  &\it sum(m, n)=sum(n, m)\\[-1ex]
  &\it sum(m, n_{\rm 1})=sum(n_{\rm 1}, m)\\[-1ex]
  &\it m = sum(0, m)\\[-1ex]
  &\it m_{\rm 1} = sum(0, m_{\rm 1})
\end{aligned}
&&
\begin{aligned}
  &\it s(m_{\rm 1}) = sum(0, s(m_{\rm 1})) \\[-1ex]
  &\it sum(m, s(n_{\rm 1})) = sum(s(n_{\rm 1}), m)
\end{aligned}
&&
\textsf{[Reduce]}\displaybreak[0]\\
\hline
5.
&
\qquad\textrm{idem}
&&
\it sum(m, s(n_{\rm 1})) = sum(s(n_{\rm 1}), m)
&&
\textsf{[Derive]}\displaybreak[0]\\
\hline
6.
&
\begin{aligned}
  &\it sum(m, n)=sum(n, m)\\[-1ex]
  &\it sum(m, n_{\rm 1})=sum(n_{\rm 1}, m)\\[-1ex]
  &\it m = sum(0, m)\\[-1ex]
  &\it m_{\rm 1} = sum(0, m_{\rm 1})\\[-1ex]
  &\it s(sum(n_{\rm 1}, m)) {=} sum(s(n_{\rm 1}), m)\\[-1ex]
  &\it s(sum(n_{\rm 1}, m_{\rm 2})) {=} sum(s(n_{\rm 1}), m_{\rm 2})
\end{aligned}
\hspace{-4ex}
&&
\begin{aligned}
  &\it s(sum(n_{\rm 1}, 0))=sum(s(n_{\rm 1}), 0)\\[-1ex]
  &\it s(sum(s(n_{\rm 1}),\! m_{\rm 2})) {=} sum(s(n_{\rm 1}),\! s(m_{\rm 2}))
\end{aligned}
\hspace{-3ex}
&&
\,\textsf{[Reduce]}\displaybreak[0]\\
\hline
7.
&
\qquad\textrm{idem}
\hspace{-4ex}
&&
\it s(sum(s(n_{\rm 1}),\! m_{\rm 2})) {=} sum(s(n_{\rm 1}),\! s(m_{\rm 2}))
\hspace{-3ex}
&&
\textsf{[Reduce]}\displaybreak[0]\\
\hline
&
\qquad\textrm{idem}
\hspace{-4ex}
&&
~~\emptyset
&&
\\
\hline
\end{align*}
\end{center}
In the first step only \textsf{[Derive]} can be applied. Two new proof obligations (goals) are derived and added to $\cal G$: the first one is obtained using the derivative $\delta_0=0$ and the second one using the derivative $\delta_s=s(N_1{:}{\it Nat})$. Since we have $n_1\sim n$, the initial goal with $n_1$ replacing $n$ is added as a frozen hypothesis to $\cal F$. In the second step we assume that $\it s(m) = sum(0, s(m))$ is processed. Then only \textsf{[Derive]} can be applied again. 
The current goal, where $m$ is replaced by $m_1$ ($m_1\sim m$) and $\it sum(m_1,0)$ is reduced to $m$, is added as a frozen hypothesis  to $\cal F$ and its derived goals are added as new proof obligations to $\cal G$. The goal $\it 0 = sum(0, 0)$ is reduced using the basic entailment relation (=  equational deduction).
In the fourth step, the goal $\it s(m_{\rm 1}) = sum(0, s(m_{\rm 1}))$ is reduced using the frozen hypothesis $\it m_{\rm 1} = sum(0, m_{\rm 1})$. 
In the fifth step, the goal $\it sum(m, s(n_{\rm 1})) = sum(s(n_{\rm 1}), m)$ is equivalent to
$\it s(sum(n_{\rm 1}, m)) = sum(s(n_{\rm 1}), m)$ by applying the definition of {\it sum} and the frozen hypothesis $\it sum(m, n_{\rm 1})=sum(n_{\rm 1}, m)$. The derived goals of this last goal are added to $\cal G$ as new proof obligations and the frozen hypothesis obtained from it by replacing $m$ by $m_2$  ($m_2\sim m_1$). 

The non-linear goals can also be handled with circular induction. For instance, the goal $\it max(N,N)=N$ can be proved just in one step:
\begin{center}
\begin{align*}
& \cal F &&\cal G &&\rm Rule\\
\hline
1.~~
&
\emptyset 
&&
\it max(n, n) = n
&&
\textsf{[Derive]}\\
\hline
2.~~
&
\begin{aligned}
  &\it max(n, n) = n\\[-1ex]
  &\it max(n_{\rm 1}, n_{\rm 1}) = n_{\rm 1}
\end{aligned}
&&
\begin{aligned}
  &\it max(0,0)=0\\[-1ex]
  &\it max(s(n_{\rm 1}), s(n_{\rm 1})) = s(n_{\rm 1})
\end{aligned}
&&
\textsf{[Reduce]}\displaybreak[0]\\
\hline
3.~~
&
\qquad\textrm{idem}
&&
\it max(s(n_{\rm 1}), s(n_{\rm 1})) = s(n_{\rm 1})
&&
\textsf{[Reduce]}\\
\hline
&
\qquad\textrm{idem}
&&
\emptyset
&&
\\
\hline
\end{align*}
\end{center}
It is known that non-linear conjectures are problematic for automated inductive provers \cite{FalkeK06}; all non-linear examples from \cite{FalkeK06} were successfully handled with circular induction.
\end{example}

\begin{theorem}\textbf{(soundness of circular induction proof system)}\\
\label{thm:soundness-ci}
Let $\cal B$ a specification with constructors and $G$ be a set of goals. 
If $({\cal B},{\cal F}_0=\emptyset,{\cal G}_0=\efr{G})\Rightarrow\ldots\Rightarrow ({\cal B},{\cal F}_n,{\cal G}_n=\emptyset)$ using the rules of the proof system defined in Definition~\ref{def:ci},  then ${\cal B}\bentails G$.
\end{theorem}
\begin{proof}
Let $H$ denote the set of unfrozen equations from ${\cal H}(\efr{e},{\it FVar}(\Delta[\efr{e}]))$.
First we show that $H$ is a set of special hypotheses for $e$. Let $\theta$ be an experiment such that 
${\cal B}\entails\theta^\leq[\efr{e}]$. We have to show that ${\cal B}\entails\theta[\efr{H}]$. Let $e'$ be an arbitrary equation from $H$. It follows that $\efr{e'}$ is $\efr{e}[y'/y]$ for some $y'\sim y$. There is a constructor $c$ such that $\theta(y)=c(\ldots \theta(y')\ldots)$ by the definition of $\sim$. Let $\theta'$ denote the substitution equal to $\theta$ excepting $\theta'(y)$ which is equal to $\theta(y')$. Obviously, $\theta'\leq \theta$ and hence  ${\cal B}\entails\theta'[\efr{e}]$. The conclusion ${\cal B}\entails\theta[\efr{e'}]$ follows from the fact that $\theta'[\efr{e}]$ and $\theta[\efr{e'}]$ are the same.\\
The rest of the proof of the theorem follows the same steps as that of Theorem 5 in \cite{DBLP:conf/icfem/LucanuR09}, but using the dual notions: ``equation without inductive variables" instead of ``visible equation" and ``equation with inductive variables" instead of ``hidden equation". 
\qed\end{proof}

\begin{example} (\textit{Failing example})\\
\label{ex:fail}
We present an example where the circular induction proof system defined in Definition~\ref{def:ci} fails. Later, we show how the proof system can be extended in order to handle such examples.
Assume that $\cal B$ is the specification {\it TREE} defined in Example \ref{ex:tree} and that $\cal G$ consists of the goal $\it mirror(mirror(\ell))=\ell$, where $\ell$ is the frozen form of $\it L{:}TList$. The initial $\cal F$ is empty. 
The derivation tree generated by the proof system defined in Definition~\ref{def:ci} is described as follows
\begin{align*}
& \cal F &&\cal G &&\rm Rule\\
\hline
1.~
&
\emptyset 
&&
\it mirror(mirror(\ell))=\ell
&&
\textsf{[Derive]}\\
\hline
2.~
&
\begin{aligned}
  &\it mirror(mirror(\ell))=\ell\\[-1ex]
  &\it mirror(mirror(\ell_{\rm 1}))=\ell_{\rm 1}\\[-1ex]
  &\it mirror(mirror(\ell_{\rm 2}))=\ell_{\rm 2}
\end{aligned}
&&
\begin{aligned}
  &\it mirror(mirror(nil))=nil\\[-1ex]
  &\it mirror(mirror(\ell_{\rm 1};\ell_{\rm 2}))=\ell_{\rm 1};\ell_{\rm 2}\\[-1ex]
  &\it mirror(mirror([t]))=[t]
\end{aligned}
&&
\textsf{[Reduce]}\displaybreak[0]\\
\hline
3.~
&
\qquad\textrm{idem}
&&
\begin{aligned}
  &\it mirror(mirror(\ell_{\rm 1};\ell_{\rm 2}))=\ell_{\rm 1};\ell_{\rm 2}\\[-1ex]
  &\it mirror(mirror([t]))=[t]\\
\end{aligned}
&&
\textsf{[Reduce]}\displaybreak[0]\\
\hline
4.~
&
\qquad\textrm{idem}
&&
\it [mirror(mirror(t))]=[t]
&&
\textsf{[Derive]}\displaybreak[0]\\
\hline
5.~
&
\begin{aligned}
  &\it mirror(mirror(\ell))=\ell\\[-1ex]
  &\it mirror(mirror(\ell_{\rm 1}))=\ell_{\rm 1}\\[-1ex]
  &\it mirror(mirror(\ell_{\rm 2}))=\ell_{\rm 2}\\[-1ex]
  &\it [mirror(mirror(t))]=[t]
\end{aligned}
\hspace{-2ex}
&&
\begin{aligned}
  &\it [mirror(mirror([e_{\rm 1}]))]=[[e_{\rm 1}]]\\[-1ex]
  &\it [mirror(mirror(tr(e_{\rm 2},\ell_{\rm 3}]))]=[tr(e_{\rm 2},\ell_{\rm 3})]
\end{aligned}
\hspace{-2ex}
&&
\textsf{[Reduce]}\displaybreak[0]\\
\hline
6.~
&
\qquad\textrm{idem}
\hspace{-2ex}
&&
\it [mirror(mirror(tr(e_{\rm 2},\ell_{\rm 3}))]=[tr(e_{\rm 2},\ell_{\rm 3})]
&&
\hspace{-2ex}
\\
\hline
\end{align*}

\noindent
where $\ell_i$ is the frozen form of $L_i{:}\it TList$, $t$ is the frozen form of $T{:}\it Tree$, and $e_i$ is the frozen form of $E_i{:}\it Elt$.
 In the first step, \textsf{[Derive]} produces two hypotheses because $\ell_1\sim\ell$ and $\ell_2\sim\ell$; these hypotheses will be used in the non-linear conjecture added by this rule.  The last goal can be reduced to 
$\it [tr(e_{\rm 2}, mirror(mirror(\ell_{\rm 3}))]=[tr(e_{\rm 2},\ell_{\rm 3})]$, which is irreducible because we have no hypotheses on $\ell_3$.
The problem could be solved if we extend $\sim$ such that $t \sim \ell$, $\ell_3\sim t$.
Then, at the fourth step when \textsf{[Derive]} is applied again, we add $\it mirror(mirror(\ell_{\rm 3}))=\ell_{\rm 3}$ to $\cal F$, which is a special hypothesis, because $\ell_3\sim^+\ell$, where $\sim^+$ is the transitive closure of~$\sim$. Next we formalize this extension. 
\end{example}

Let $\delta_c$ be a derivative and $e$ a goal.
If $\delta_c[\efr{e}]$ is an equation of the form $\efr{e}[c(y_1{.}s_1,\ldots,y_n{.}s_n)/y{.}s]$, then $y_i{.}s_i \simstar  y{.}s$ for all $i=1,\ldots, n$. Let $\simstarplus$ denote the transitive closure of $\simstar$.

\begin{definition}
\label{def:ciext}
The \textbf{extended circular induction proof system} consists of the following inference rules, represented as conditional reduction rules:
\begin{align*}
& {\sf [Reduce]}\\
&\quad ({\cal B},{\cal F},{\cal G}\cup\{\efr{e}\})\Rightarrow ({\cal B},{\cal F},{\cal G})\quad{\rm~if~}{\cal B}\cup{\cal F}\entails \efr{e}\displaybreak[0]\\
& {\sf [Derive]}\\
&\quad ({\cal B},{\cal F},{\cal G}\cup\{\efr{e}\})\Rightarrow\\
&\quad ({\cal B},{\cal F}\cup\{\efr{e}\}\cup{\cal H}(\efr{e},{\it FVar}(\Delta[\efr{e}]))\cup{\cal H}^*({\cal F},{\it FVar}(\Delta[\efr{e}])),{\cal G}\cup\Delta[\efr{e}])\\
&\quad {\rm~if~}{\cal B}\cup{\cal F}\not\entails \efr{e}\\[-6ex]
\end{align*}
where ${\cal H}^*(\efr{f}, FV)=\{\efr{f}[y'/y]\mid y'\simstarplus y, y'\in FV, y\in{\it FVar}(\efr{f}), {\it sort}(y)={\it sort}(y')\}$ and  ${\cal H}^*({\cal F}, FV)=\bigcup_{\efr{e}\in {\cal F}} {\cal H}^*(\efr{e}, FV)$
\end{definition}

${\cal H}^*({\cal F},{\it FVar}(\Delta[\efr{e}]))$ includes an instance of each frozen hypothesis in $\cal F$ having a frozen variable $y$ for which there is a similar new frozen variable $y'$, added by the new derived goals. Using $\simstarplus$, the "reincarnation" of a variable $y$ can be made by means of some intermediate "hosts" (= variables of different sorts).

\begin{example}
\label{ex:succ}
Using the rules of the extended circular induction proof system, the derivation tree from Example~\ref{ex:fail} can be continued as follows:
\begin{align*}
\hline
4.~
&
\begin{aligned}
  &\it mirror(mirror(\ell))=\ell\\[-1ex]
  &\it mirror(mirror(\ell_{\rm 1}))=\ell_{\rm 1}\\[-1ex]
  &\it mirror(mirror(\ell_{\rm 2}))=\ell_{\rm 2}
\end{aligned}
&&
\it [mirror(mirror(t))]=[t]
&&
\textsf{[Derive]}\displaybreak[0]\\
\hline
5.~
&
\begin{aligned}
  &\it mirror(mirror(\ell))=\ell\\[-1ex]
  &\it mirror(mirror(\ell_{\rm 1}))=\ell_{\rm 1}\\[-1ex]
  &\it mirror(mirror(\ell_{\rm 2}))=\ell_{\rm 2}\\[-1ex]
  &\it mirror(mirror(\ell_{\rm 3}))=\ell_{\rm 3}\\[-1ex]
  &\it [mirror(mirror(t))]=[t]
\end{aligned}
&&
\hspace{-2ex}
\begin{aligned}
  &\it [mirror(mirror([e_{\rm 1}]))]=[[e_{\rm 1}]]\\[-1ex]
  &\it [mirror(mirror(tr(e_{\rm 2},\ell_{\rm 3}]))]=[tr(e_{\rm 2},\ell_{\rm 3})]
\end{aligned}
&&
\hspace{-2ex}
\textsf{[Reduce]}\displaybreak[0]\\
\hline
6.~
&
\qquad\textrm{idem}
&&
\hspace{-2ex}
\it [mirror(mirror(tr(e_{\rm 2},\ell_{\rm 3}))]=[tr(e_{\rm 2},\ell_{\rm 3})]
&&
\hspace{-2ex}
\textsf{[Reduce]}\displaybreak[0]\\
\hline
7.~
&
\qquad\textrm{idem}
&&
\emptyset
&&
\hspace{-2ex}
\\
\hline
\end{align*}

\noindent
Comparing with the previous proving attempt, a new frozen axiom was added to $\cal F$ in the fourth step because $\ell_3\simstarplus\ell$. 
\end{example}


\begin{theorem}\textbf{(soundness of extended circular induction proof system)}\\
\label{thm:soundness-ci-ext}
Let $\cal B$ a specification with constructors and $G$ be a set of goals. 
If $({\cal B},{\cal F}_0=\emptyset,{\cal G}_0=\efr{G})\Rightarrow\ldots\Rightarrow ({\cal B},{\cal F}_n,{\cal G}_n=\emptyset)$ using the rules of the proof system defined in Definition~\ref{def:ciext},  then ${\cal B}\bentails G$.
\end{theorem}

\begin{proof}
It is enough to observe that ${\cal H}^*({\cal F},{\it FVar}(\Delta[\efr{e}]))$ is a set of special hypotheses for $\cal F$. The rest of the proof is similar to that of Theorem~\ref{thm:soundness-ci}.
\qed\end{proof}

\section{Implementation in CIRC}
\label{sec:impl}

{\CIRC} \cite{circ-zenodo,lucanu-etal-2009-calco} is an automated inductive and coinductive prover built around the circularity principle.
It is implemented in Maude \cite{DBLP:conf/maude/2007} as a metalanguage application, thus making
use of important reflection capabilities in order to enhance a
specification during a proving session. The prover extends Maude specifications
with behavioral properties. The specifications provided to {\CIRC} are order-sorted \cite{osa1} rather than
many-sorted. This involves the definition of a partial order relation
between sorts which has the role of avoiding the explicit use of
inclusion operators. For instance, for the specification presented
in Example~\ref{ex:tree}, instead of using the operation
$[\_]:\it Tree \to TList$
we can simply specify the subsort relation $\it Tree < TList$.
The capability of working with subsorts is also exploited in the efficient implementation of the circular induction. 
 
As presented in \cite{lucanu-etal-2009-calco}, the entailment relation used in \CIRC\ is
$\mathcal B\circentails (\forall X)t=t'$ {\tt if} $\land_{i\in I} (u_i=v_i)$ iff
${\rm nf}(t)={\rm nf}(t')$, where ${\rm nf}(t)$ is computed
as follows:\\
\indent -- the variables $X$ of the equations are turned into fresh
constants;\\
\indent -- the  condition equalities $u_i=v_i$ are added as equations to the
specification;\\
\indent -- the equations in the specification are oriented and used as rewrite
rules.

The engine implements the proof system
given in \cite{rosu-lucanu-2009-calco} by a set of reduction rules
$({\cal B}, {\cal F}, {\cal G}) \Rightarrow ({\cal B}, {\cal F'}, {\cal G'})$, where
${\cal B}$ represents the (original) algebraic specification,
${\cal F}$ is the set of frozen axioms and
${\cal G}$ is the current set of proof obligations. During a proof by induction
session, the steps followed by the engine resemble
those made during a proof by coinduction. The difference consists of the way
in which frozen goals and special hypotheses are handled.

In what follows we present the correlation between the theory presented so far
and the way  {\CIRC} manages data during a proving session by induction. 
The inference rules are implemented using a similar mechanism to that used for coinductive proofs. 
The main difficulty raised by circular  induction 
is the internal management of special hypotheses
${\cal H}({\cal F},{\it FV}) \cup
{\cal H}^*({\cal F},{\it FV})$. The number of these hypotheses could exponentially increase, and their computation become expensive. We use the subsort facility in order to drastically reduce the number of stored hypotheses. We store only generic hypotheses represented with variables of fresh subsorts and each new generated constant is declared as inhabitant to an appropriate such subsort. Then the set of all needed special hypotheses are obtained by instantiation. For that we need to distinguish between initial variables and their incarnations.

\begin{definition}
A frozen inductive variable $y.s$ is a \textbf{source} iff
$({\not \exists} y'.s)  y.s \simstarplus  y'.s$.
\end{definition}

For each frozen induction variable $y'.s$ we define the operation {\it source}
such that ${\it source}(y'.s) = y.s$ if $y'.s \simstarplus y.s$
and $y.s$ is a source. Each time we add a new source $y.s$, either at
the beginning, or during a proof session, we enrich
the specification with a fresh subsort $s\langle y \rangle$ such that $s\langle y \rangle < s$.

{\CIRC} implements a variant of the rule \textsf{[Derive]}
from Definition~\ref{def:ciext} such that instead of adding a specific
hypothesis (equation) for each of the generated variables, it adds only one
hypothesis that may be used by all the frozen variables with a common source.
This leaves the set of hypotheses uncluttered and helps for improving
the speed of the matching engine:
\begin{align*}
& {\sf [Derive]}\\
&\quad ({\cal B},{\cal F},{\cal G}\cup\{\efr{e}\})\Rightarrow
\quad ({\cal B'}, {\cal F}\cup{\cal H'}(\efr{e}, {\it FVar}(\Delta[\efr{e}])) , {\cal G}\cup\Delta[\efr{e}])
\quad {\rm~if~}{\cal B}\cup{\cal F}\not\entails \efr{e}\displaybreak[0]\\
&\,\,{\rm where~}
{\cal H}'(\efr{e}, {\it FV}) =
\{\efr{e}[y'{:}s\langle {\it source}(y) \rangle/y.s ] \mid
y'.s \in {\it FV}, y.s \in {\it FVar}(\efr{e}),\\
&\,\, y'.s\simstarplus y.s \} \textnormal{ and }
{ \cal B }' {~\rm is~} { \cal B}
\textnormal{ enriched with the constants } y'.s' \textnormal{ of sort }
s'\langle {\it source}(y')\rangle \\
&\,\,\textnormal{for each } y'.s' \in {\it FVar}(\Delta[\efr{e}]).
\end{align*}

\noindent
We have $\efr{f}\in {\cal H}^*({\cal F},{\it FV})\cup{\cal H}(\efr{e},{\it FVar}(\Delta[\efr{e}]))$ iff 
$\efr{f}$ is an instance of a hypothesis in ${\cal F}\cup{\cal H'}(\efr{e}, {\it FVar}(\Delta[\efr{e}]))$.
In order to understand the special hypothesis management algorithm better,
let us analyze the derivation steps from Examples~\ref{ex:fail}
and \ref{ex:succ}.

Step 1 involves the creation of a new sort
${\it TList}\langle\ell\rangle < {\it TList}$.
The {\it source} function is defined as:
${\it source}(\ell) = \ell$,
${\it source}(\ell_1) = \ell$, ${\it source}(\ell_2) = \ell$.
Here $\ell, \ell_1, \ell_2$ are added as constants of sort ${\it TList}\langle\ell\rangle$.
The frozen hypothesis added to the specification is
$\it mirror(mirror(L{:}TList\langle\ell\rangle)) = L{:}TList\langle\ell\rangle$.

At step 4 we add the sort ${\it Tree}\langle t\rangle < {\it Tree}$ and
enhance the {\it source} function such that
${\it source}(t) = t$, ${\it source}(l_3) = l$.
The added constants are 
$t.{\it Tree}\langle t\rangle$ and
$\ell_3.{\it TList}\langle \ell \rangle$. The new frozen
hypothesis is
$\it [mirror(mirror(T{:}Tree\langle t \rangle))] = 
[T{:}Tree\langle t \rangle]$.

The reduction at step 6 succeeds because it makes
use of the first hypothesis instantiated with $\ell_3$.

Excepting the management of the special contexts and the meanings of the frozen equations, the engines implementing circular induction and circular coinduction, respectively, are similar. An immediate main advantage is that the mechanisms added to increase the power of one engine can be used for free by the other one.
We exemplify this feature showing how the simplification rules \cite{synasc09}, initially added for the coinduction engine,  
can be used by the induction engine for deriving simpler goals.
Let us consider the goal $\sum_{k=0}^N fib(k)=np(s(N))$ from Example~\ref{ex:fib}. The sum can be computed by an operation $\it sum\langle k\rangle$ defined by $\it sum\langle k\rangle(0) =0$ and $\it sum\langle k\rangle(s(N))=sum\langle k\rangle(N)+fib(s(N))$.  
One intermediate goal computed by circular induction is
$\it fib(s(N)) + sum\langle k\rangle(N) = s(np(N))+ sum\langle k\rangle(N)$. {\CIRC} succeeds to automatically prove this goal, but it would be nice to have it in a simpler way, as a general result: $\it fib(s(N)) = s(np(N))$.
This can be accomplished by adding to the specification the following simplification rule:
\[
\dfrac{M+P=N+P}{M=N}
\]
and to define a proof tactic which tries to apply a simplification before each derivation. The proof tactics are specified in {\CIRC} using a simple strategy language~\cite{CaltaisGLG08}. 

\subsection{Circular Coinduction and Circular Induction Working Together}

Since both circular coinduction and circular induction are defined using the same algebraic framework, they could be combined for proving a given goal.
Here we present a simple example  to exhibit how the two proof techniques are combined in the current version of \CIRC. We consider the specification \textit{NATSUM} defined in Example~\ref{ex:nat-sum} and streams defined over natural numbers with the operations {\it hd}, {\it tl}, and {\it zeros} defined as in Example~\ref{ex:streams} and together with an addition operation $\_+\_:\it Stream~Stream\to Stream$ defined by $\it hd(S_{\rm 1}+S_{\rm 2})=sum(hd(S_{\rm 1}), hd(S_{\rm 2}))$, $\it tl(S_{\rm 1}+S_{\rm 2})=tl(S_{\rm 1}) + tl(S_{\rm 2})$. 
The goal is to prove that $\it zeros +S=S$.
We start by applying the circular coinduction proof tactic:

{\fontsize{9}{9}\selectfont
\begin{verbatim}
(add goal zeros + S:Stream = S:Stream .)
(coinduction .)
. . .
Visible goal [* sum(0,hd(S:Stream)) *] = [* hd(S:Stream) *] 
failed during coinduction.
\end{verbatim}
}

\noindent
The proof sticks on the intermediate goal $\it sum(0,hd(S)) = hd(S)$, because that it is not direct consequence of {\it NATSUM}, and we have to prove it by induction. We cannot apply the circular induction proof technique because the goal does not include a variable of sort {\it Nat}. We save the current state of the proof, and start a proof by induction on a generalization of the goal: $\it sum(0, N)=N$, where $N$ is a variable of sort {\it Nat} (this can be obtained using the command \texttt{(generalize~.)}). After the engine finishes the proof by induction, it automatically reloads the saved state and resume the left proof by coinduction:

{\fontsize{9}{9}\selectfont
\begin{verbatim}
(save proof state .)
(add goal sum(0, N:Nat) = N:Nat .)
(induction  .)
. . .
Proof state loaded.
(coinduction .)
Proof succeeded.
. . .
Proved properties:
  sum(0,N#0:Nat<N>) = N#0:Nat<N>
  sum(0,N:Nat) = N:Nat
  zeros + S:Stream = S:Stream
\end{verbatim}
}

\noindent
The above procedure is not yet completely mechanized. In order to make it automated, the tool must be extended with the ability to handle the induction reasoning as a basic entailment relation for the circular coinduction.

\section{Related Work}
\label{sec:related-work}

There are many mechanizable induction techniques proposed in the
literature; reviews of the explicit techniques can be found
in~\cite{ci:bundy} and of the implicit ones in~\cite{ci:comon}. We
group the discussion into two families: the classical
coverset/test-set tradition, to which circular induction was
originally compared, and the more recent \emph{cyclic-proof} (or
\emph{infinite-descent}) tradition, into which circular induction most
naturally fits when viewed with hindsight.

\subsection{Coverset, rewrite and test-set induction}

Coversets provide the basis for many mechanizable induction
principles. Coverset induction (see, e.g.,~\cite{ci:zhang}) explores
the recursive definitions of the operations occurring in conjectures;
the method is tied to a well-founded ordering, which is not always easy
to find. The set of derivatives used in circular induction are
coversets associated with \emph{sorts} (not operations), used to derive
new goals that are not always subgoals; circular induction requires no
explicit well-founded ordering, its soundness resting instead on the
implicit ordering of the experiments. Rewrite (term-rewriting)
induction~\cite{ci:reddy} is also based on coversets and uses the
rewrite relation of a uniformly terminating rewrite system as its
well-founded ordering, which typically assumes an explicit partial
ordering over operators. Test-set induction (see,
e.g.,~\cite{ci:bouhoula97}) combines an explicit, test-set--based
induction scheme with proof by consistency (implicit induction) used as
a refutation technique, and was implemented in
\textsc{Spike}~\cite{ci:spike}. Circular induction discovers the
analogue of a test set on-the-fly while it unrolls the derivation
graph, with no separate algorithm for computing test/cover sets and no
built-in refutation procedure.

\subsection{Cyclic proof and proof by infinite descent}

Seen from today, circular induction is an instance of a much broader
idea that crystallized independently and almost simultaneously in the
proof-theory community: \emph{cyclic proof}, in which the explicit
induction rule is replaced by a finite derivation containing
``back-links'' (cycles), whose soundness is guaranteed by a global
condition expressing an argument by Fermat-style \emph{infinite
descent}. The foundational reference is Brotherston and
Simpson~\cite{ci:brotherston-simpson}, who give two sequent calculi for
first-order logic with inductively defined predicates: one with
explicit (Henkin-complete) induction rules, and one with
\emph{non-well-founded} proofs in which the left rules for inductive
predicates are mere case-splits and soundness is recovered by a global
\emph{trace condition}---every infinite path must carry a trace that
progresses (decreases) infinitely often. Restricting the infinitary
system to regular (finitely representable) proofs yields the
\emph{cyclic} system, which they show subsumes explicit induction and
conjecture to be equivalent to it. Circular induction can be read as
the dual, \emph{equational-logic} counterpart of this picture: our
\emph{derivatives} are exactly the case-split (unfolding) steps; our
\emph{frozen equations} play the structural role that the trace
condition plays globally, inhibiting unsound use of a hypothesis until
descent has genuinely occurred; and our soundness theorems
(Theorems~\ref{thm:icp}, \ref{thm:soundness-ci} and \ref{thm:soundness-ci-ext}) appeals to the implicit well-founded order on
experiments rather than to a separately checked global condition. A
distinctive feature of our setting is that the very same circularity
principle, applied to destructors instead of constructors, yields
circular \emph{co}induction, so that induction and coinduction live in
one framework---a unification that the cyclic-proof literature reached
only later (see below).

\paragraph{Generic cyclic provers (\textsc{Cyclist}).}
Brotherston, Gorogiannis and Petersen~\cite{ci:cyclist} distil the
construction of~\cite{ci:brotherston-simpson} into \textsc{Cyclist}, a
generic, fully automated cyclic theorem prover: an OCaml functor
parametric in the sequents, inference rules and trace-pair function of
the object logic, instantiated to first-order logic with inductive
predicates, separation-logic entailment, and separation-logic program
termination. The global trace condition is decided by checking
language inclusion of two B\"uchi automata via a model checker.
\textsc{Cyclist} and CIRC share the ambition of a single, logic-generic
engine---\textsc{Cyclist} is parametric in the logic, while circular
(co)induction is parametric in the basic entailment relation $\vdash$.
The technical contrast is instructive. \textsc{Cyclist} reasons about
inductively defined \emph{predicates} and, as its authors note, has
``difficulty dealing with heavily-equational goals''; commutativity of
addition is reported as not provable, precisely because it needs
equational rewriting and a non-trivial lemma. Circular induction,
working natively over equational specifications, proves commutativity
and associativity of \texttt{sum} automatically, discovering the
required lemmas as frozen hypotheses. Conversely, \textsc{Cyclist}'s
B\"uchi-automaton soundness check copes with overlapping cycles and
complex induction schemes (e.g.\ Wirth's ``P\&Q'' example) that are
beyond a purely local discipline; circular induction trades some of
that generality for the simplicity of needing only two inference rules
and no global automaton check.

\paragraph{Synthesising the induction (cyclic abduction).}
Brotherston and Gorogiannis~\cite{ci:cyclic-abduction} push cyclic
proof from \emph{verification} towards \emph{synthesis}: their
\textsc{Caber} tool (built on \textsc{Cyclist}) performs \emph{cyclic
abduction}, simultaneously constructing a cyclic safety/termination
proof of a heap-manipulating \texttt{while} program and \emph{abducing}
the inductive separation-logic predicates (lists, trees, cyclic and
composite structures) that the proof needs. This is a close relative,
in spirit, of the lemma discovery that circular induction
performs while it unrolls a proof: in both cases the missing inductive
content is guessed on demand to make the cycle close. The settings
differ---\textsc{Caber} abduces shape predicates for imperative pointer
programs in separation logic, whereas circular induction discovers
auxiliary equational lemmas (e.g.\ $\mathit{sum}(0,M)=M$) and
incarnation relations for algebraic data types.

\paragraph{Integrating induction and coinduction.}
Cohen and Rowe~\cite{ci:cohen-rowe} are, to our knowledge, the first in
the cyclic-proof tradition to integrate induction and coinduction in
one system, an aim that motivated circular (co)induction from the
outset. They extend transitive-closure logic with a dual
``co-closure'' operator forming \emph{greatest} fixed points, and give
a sound and complete non-well-founded proof system whose cyclic
subsystem supports automated inductive \emph{and} coinductive reasoning,
demonstrated on streams. The parallel with our Figure~\ref{fig:duality}
is striking: where we obtain coinduction by dualising constructors to
destructors and ``freezing at top'' to ``inhibiting at bottom'', they
obtain it by dualizing least to the greatest fixed points via closure/%
co-closure. Their development is carried out for a relational fixed-point
logic with completeness results; ours is carried out for behavioral
equational logic with an implemented prover (CIRC), and the two can be
seen as complementary realizations of the same unifying intuition.

\paragraph{Deciding infinite descent efficiently (\textsc{Cyclone}).}
A recurring cost in the \textsc{Cyclist} line is the soundness check
itself: deciding the infinite-descent/trace condition is
PSPACE-complete. Cohen, Rowe and Shaked~\cite{ci:cyclone} address this
with \textsc{Cyclone}, a \emph{heterogeneous} tool that combines several
fast semi-decision procedures with a complete decision
procedure as a fallback, achieving high coverage and large speed-ups on
benchmarks harvested from \textsc{Cyclist}. This work is orthogonal to,
but illuminating for, circular induction: it engineers exactly the
global check that our use of frozen equations and the implicit ordering
of experiments lets us \emph{avoid}, at the price of the more
restricted (but cheaper) local discipline we adopt.

\paragraph{Cyclic equational reasoning (\textsc{CycleQ}).}
The closest modern analogue of circular induction is arguably
\textsc{CycleQ} by Jones, Ong and Ramsay~\cite{ci:cycleq}, a cyclic
proof system for automated \emph{equational} reasoning about pure
functional programs, in which cyclic proof and equational reasoning are
mediated by contextual substitution as a cut rule. Like circular
induction, \textsc{CycleQ} replaces the explicit induction rule by
cyclic unfolding over an equational goal, and the authors explicitly
position it as subsuming ``inductionless induction'', ``rewriting
induction'' and ``proof by consistency''---the same classical lineage
to which we compared circular induction above. Its soundness is checked
incrementally with the \emph{size-change principle} (rather than a
B\"uchi-automaton inclusion test), and it is implemented as a GHC
plugin. The conceptual overlap with our work is considerable; the main
differences are that \textsc{CycleQ} targets a typed functional language
with a size-change soundness check and contextual-substitution cut,
whereas circular induction targets order-sorted algebraic
specifications, derives soundness from the experiment ordering, and is
embedded in the same engine as circular coinduction.

\paragraph{Inductive theorem proving in Maude (\textsc{NuITP}).}
Within the same Maude/rewriting-logic ecosystem as CIRC, Dur\'an,
Escobar, Meseguer and Sapi\~na~\cite{ci:nuitp} have developed
\textsc{NuITP}, a next-generation inductive theorem prover based on
inductive order-sorted first-order logic. \textsc{NuITP} targets the
same equational programs as circular induction but emphasises
expressiveness---types and subtypes, conditional equations and
rewriting modulo associativity, commutativity and identity---and
scalability, leveraging Maude's modern symbolic infrastructure (variant
narrowing and unification, variant satisfiability, order-sorted
congruence closure, equality predicates). Its central inference rule,
Generator-Set Induction, is an explicit, constructor-based induction
scheme: it requires a user-supplied recursive path ordering (RPO) for
termination and manages explicit induction hypotheses, in contrast to
circular induction, which uses no explicit well-founded ordering and
discovers its hypotheses (frozen equations) during proof search. The
two thus represent the explicit-induction and circular-induction
strategies, respectively, instantiated over very similar order-sorted
equational theories, and the comparison highlights what each style buys:
\textsc{NuITP} gains modulo-AC(U) reasoning, decision procedures and
scalability; circular induction gains a minimal two-rule calculus and
free interoperation with circular coinduction.

\begin{table}[th]
\centering
\caption{Circular induction in the context of cyclic / infinite-descent
provers. ``Implicit order'' means soundness rests on a well-founded
order that is never computed explicitly.}
\label{tab:cyclic-comparison}
\small
\setlength{\tabcolsep}{4pt}
\begin{tabular}{@{}p{2.6cm}p{2.1cm}p{2.8cm}p{2.6cm}p{1.4cm}p{2.2cm}@{}}
\toprule
\textbf{System/paper} &
\textbf{Paradigm} &
\textbf{Logic/domain} &
\textbf{Soundness mechanism} &
\textbf{Explicit order?} &
\textbf{Ind.+Coind.?} \\
\midrule
\textbf{Circular induction} (this paper, 2010) &
Circular (cyclic) &
Order-sorted equational specs &
Frozen eqs.\ + implicit order on experiments &
No &
Yes (with circular coind.) \\
\addlinespace
Brotherston--Simpson \cite{ci:brotherston-simpson} (LICS'07) &
Cyclic / explicit (both) &
FO logic + inductive predicates &
Global trace condition (infinite descent) &
Implicit (trace) &
Induction only \\
\addlinespace
\textsc{Cyclist} \cite{ci:cyclist} (APLAS'12) &
Cyclic, generic &
FO ind.\ preds.; sep.\ logic; termination &
Trace via B\"uchi-automata inclusion &
Implicit (trace) &
Induction only \\
\addlinespace
\text{Cyclic~abduction} /\;\textsc{Caber} \cite{ci:cyclic-abduction} (SAS'14) &
Cyclic + abduction (synthesis) &
Separation logic, \texttt{while} programs &
Safety/termination trace condition &
Implicit (trace) &
Induction only \\
\addlinespace
Cohen--Rowe \cite{ci:cohen-rowe} (IJCAR'20) &
Non-well-founded / cyclic &
Transitive (co-)closure logic &
Trace condition; sound \& complete &
Implicit (trace) &
\textbf{Yes} (closure/co-closure) \\
\addlinespace
\textsc{Cyclone} \cite{ci:cyclone} (TACAS'25) &
Decision procedures for cyclic proofs &
Trace graphs (logic-agnostic) &
Heterogeneous (semi-)decision of descent &
n/a &
n/a \\
\addlinespace
\textsc{CycleQ} \cite{ci:cycleq} (PLDI'22) &
Cyclic, equational &
Typed pure functional programs &
Size-change principle; contextual-subst.\ cut &
Implicit (size-change) &
Induction only \\
\addlinespace
\textsc{NuITP} \cite{ci:nuitp} (PPDP'24\,/\,2025) &
Explicit induction &
Order-sorted eq.\ logic, modulo AC(U) &
User-supplied RPO + induction hyps. &
\textbf{Yes} (RPO) &
Induction only \\
\bottomrule
\end{tabular}
\end{table}

\subsection{Summary}

Table~\ref{tab:cyclic-comparison} summarises the comparison along the
dimensions that matter most for circular induction: the proof paradigm,
the logic/domain, the mechanism that justifies soundness, whether an
explicit well-founded order is required, what (if anything) is
discovered during the proof, whether the framework is generic, and
whether induction and coinduction are unified.

\section{Conclusion}
\label{sec:concl}

We have shown that the Circularity Principle, originally used to justify
circular coinduction, also underlies a sound inductive proving technique,
circular induction. Presenting the two techniques at a generic level,
parametric in the basic entailment relation, exposes them as dual
instances of one principle: derivatives arise from constructors rather
than destructors, experiments are ground substitutions rather than
contexts, and freezing inhibits the use of a hypothesis ``over'' a
substitution rather than ``underneath'' a context. This duality let us
reuse much of the metatheory of circular coinduction and, at the level of
implementation, let the two engines in CIRC share machinery---so that
enhancements such as simplification rules, generalization, or special
hypotheses benefit both. The resulting calculus is deliberately minimal:
two inference rules, no explicit well-founded ordering, and no separate
algorithm for computing cover or test sets, with the required inductive
content (auxiliary lemmas and the analogue of a test set) discovered on
the fly as the derivation graph is unrolled.

Because coinductive behavioral equality and inductive equality share the
same abstract structure based on experiments, they also share the same
complexity: inductive equality is $\Pi^0_2$-complete, so no procedure can
be complete for proving \emph{and} disproving inductive conjectures. In
particular, the present version of circular induction has no explicit
refutation mechanism, and a non-terminating derivation may indicate either
that the goal is false or merely that the system was unable to prove it.

Several directions remain open. The combination of induction and
coinduction illustrated in Section~5 is not yet fully mechanized: to
automate it, the tool must be able to use inductive reasoning as a basic
entailment relation for circular coinduction. Adding a complementary
disproving or refutation procedure would let circular induction reject
false conjectures rather than merely fail to prove them.


\end{document}